\documentclass{article}
\usepackage{ijcai23}

\usepackage{times}
\usepackage{soul}
\usepackage{url}
\usepackage[hidelinks]{hyperref}
\usepackage[utf8]{inputenc}
\usepackage[small]{caption}
\usepackage{graphicx}
\usepackage{tikz}
\usepackage{pgfplots}
\pgfplotsset{width=10cm,compat=1.9}
\usepackage{amsmath}
\usepackage{amsthm}
\usepackage{booktabs}
\usepackage{algorithm}
\usepackage{algorithmic}
\usepackage[switch]{lineno}

\allowdisplaybreaks
\usepackage{amssymb,amsmath,amsthm}
\usepackage{bbm}
\usepackage[capitalise]{cleveref}

\newtheorem{lemma}{Lemma}
\newtheorem{theorem}[lemma]{Theorem}

\newtheorem{cor}[lemma]{Corollary}

\newcommand{\E}{\mathbb{E}}
\newcommand{\R}{\mathbb{R}}

\usepackage{xcolor}

\title{Outsourcing Adjudication to Strategic Jurors}

\author{
    Ioannis Caragiannis \and Nikolaj I. Schwartzbach\thanks{This project is funded by VILLUM FONDEN under the Villum Kann Rasmussen Annual Award in Science and Technology under grant agreement no 17911.}
    \affiliations
    Department of Computer Science, Aarhus University
    \emails
    \{iannis,nis\}@cs.au.dk
}

\begin{document}

\maketitle

\begin{abstract}
We study a scenario where an adjudication task (e.g., the resolution of a binary dispute) is outsourced to a set of agents who are appointed as jurors. This scenario is particularly relevant in a Web3 environment, where no verification of the adjudication outcome is possible, and the appointed agents are, in principle, indifferent to the final verdict. We consider simple adjudication mechanisms that use (1) majority voting to decide the final verdict and (2) a payment function to reward the agents with the majority vote and possibly punish the ones in the minority. Agents interact with such a mechanism strategically: they exert some effort to understand how to properly judge the dispute and cast a yes/no vote that depends on this understanding and on information they have about the rest of the votes. Eventually, they vote so that their utility (i.e., their payment from the mechanism minus the cost due to their effort) is maximized. Under reasonable assumptions about how an agent's effort is related to her understanding of the dispute, we show that appropriate payment functions can be used to recover the correct adjudication outcome with high probability. Our findings follow from a detailed analysis of the induced strategic game and make use of both theoretical arguments and simulation experiments.
\end{abstract}

\section{Introduction}
We consider the problem of incentivizing jurors to properly assess case evidence, so that the resulting adjudication is better than random. The problem is motivated by dispute resolution in Web3 systems, where a reliable solution would find numerous applications in, e.g., supply chain management, banking, and commerce \cite{icbc}.

Web3 typically assumes no trusted authorities and adjudication must therefore be delegated to ordinary users (or agents), who are appointed as jurors and get compensated for this activity. Such agents are anonymous and cannot easily be held accountable for their actions. They are largely indifferent to the outcome of the adjudication case and typically strategize to maximize their utility. As such, paying a fixed reward to the agents for their participation is insufficient; they can then just vote randomly, without putting in any effort to assess the case evidence, producing a useless adjudication outcome. Instead, to produce a non-trivial adjudication, payments to/from the agents should be in some way conditioned on their vote. Hopefully, if the agents are satisfied with their payments, they will make a reasonable effort to assess the case evidence and collectively come up with a correct adjudication. We ask the following natural question.

\begin{quote}
    \em How can payments be structured to motivate strategic jurors to collectively produce a correct adjudication when they are indifferent to the outcome?
\end{quote}

\noindent We consider binary (yes/no) adjudication tasks and the following simple mechanism. Each agent submits a vote with her opinion and the adjudication outcome is decided using majority. Agents are rewarded for voting in accordance with the final verdict and less so for voting otherwise. This approach has been deployed in real systems like Kleros~\cite{kleros_old,kleros}. Kleros is already deployed on Ethereum and, at the time of writing, it has allegedly settled more than one thousand disputes.

\paragraph{Our contributions.}
Our main conceptual contribution is a new model for the behaviour of strategic agents. The model aims to capture the two important components of strategic behaviour while participating in an adjudication task. The first one is to decide the effort the agent needs to exert to get sufficient understanding of the task and form her opinion. The second one is whether she will cast this opinion as vote or she will vote for the opposite alternative. We assume that, when dealing with an adjudication task, agents do not communicate with each other. Instead, each of them has access to the outcome of similar tasks from the past. An agent can compare these outcomes to her own reasoning for them, which allows her to conclude whether her background knowledge is positively correlated, negatively correlated, or uncorrelated to the votes cast by the other agents. Payments can be used to amplify the agent's incentive to take such correlation into account. A strategic agent then acts as follows. If there is positive correlation, her opinion for the new adjudication task will be cast as vote. If correlation is negative, she will cast the opposite vote. If there is no correlation, the agent will vote randomly.

We assume that each adjudication task has a ground truth alternative that we wish to recover. Agents are distinguished into well-informed and misinformed ones. Well-informed (respectively, misinformed) agents are those whose opinions get closer to (respectively, further away from) the ground truth with increased effort. The ground truth is unobservable and, thus, the agents are not aware of the category to which they belong.

After presenting the strategic agent model, we characterize the strategies of the agents at equilibria of the induced game. We use this characterization to identify a sufficient condition for payments so that equilibria are simple, in the sense that the agents either vote randomly or they are all biased towards the same alternative. Next, we focus on a simple scenario with a population of well-informed and misinformed agents with complementary effort functions and show how to efficiently find payments that result in adjudication that recovers the ground truth with a given probability. Finally, we conduct experiments to justify that strategic play of a population with a majority of well-informed agents results in correct adjudication when payments are set appropriately.

\paragraph{Related work.}
Voting, the main tool we use for adjudication, has received enormous attention in the social choice theory literature ---originating with the seminal work of~\citeauthor{A51}~\shortcite{A51}--- and its recent computational treatment~\cite{comsoc}. However, the main assumption there is that agents have preferences about the alternatives and thus an interest for the voting outcome, in contrast to our case where agents' interest for the final outcome depends only on whether this gives them compensation or not. Strategic voter behaviour is well-known to alter the intended outcome of all voting rules besides two-alternative majority voting and dictatorships~\cite{G73,S75}. Positive results are possible with the less popular approach of introducing payments to the voting process; e.g., see \citeauthor{PW18}~\shortcite{PW18}.

The assumption for a ground truth alternative has been also inspired from voting theory~\cite{CPS16,CS05,Y88}. In a quite popular approach, votes are considered as noisy estimates of an underlying ground truth; typically, agents tend to inherit the preferences in the ground truth more often than the opposite ones. Our assumption for a majority of well-informed agents is in accordance with this. However, an important feature here is that the ground truth is unobservable. This is a typical assumption in the area of peer prediction mechanisms for unverifiable information (see~\citeauthor{FR17}~\shortcite{FR17}, Chapter 3), where a set of agents are used to decide about the quality of data. However, that line of work has a mechanism design flavour and assumes compensations to the agents so that their evaluation of the available data is truthful (e.g., see~ \citeauthor{Witkowski18}~\shortcite{Witkowski18}). This is significantly different than our modeling assumptions here. In particular, any evaluation of the quality of the agents---a task that is usually part of crowdsourcing systems; e.g., see~\citeauthor{SZP15}~\shortcite{SZP15}---is in our case infeasible. Still, our payment optimization is similar in spirit to automated mechanism design~\cite{S03} but, instead of aiming for truthful agent behaviour, we have a particular equilibrium as target.

\section{Modeling assumptions and notation}
We assume that adjudication tasks with two alternatives are outsourced to $n$ agents. We use the integers in $[n]=\{1, 2, ..., n\}$ to identify the agents. For an adjudication task, each agent casts a vote for one of the alternatives and the majority of votes defines the adjudication outcome. In the case of a tie, an outcome is sampled uniformly at random. To motivate voting, payments are used. A {\em payment function} $p:[0,1]\rightarrow \R$ indicates that agent $i$ gets a payment of $p(x)$ when the total fraction of agents casting the same vote as $i$ is $x$. Payments can be positive or negative (corresponding to monetary transfers to and from the agents, respectively).

The objective of an adjudication task is to recover the underlying \emph{ground truth}. We denote by $T$ the ground truth and by $F$ the other alternative. We use the terms $T$-vote and $F$-vote to refer to a vote for alternative $T$ and $F$, respectively. To decide which vote to cast, agents put an effort to understand the adjudication case and get a \emph{signal} of whether the correct adjudication outcome is $T$ or $F$. We partition the agents into two categories, depending on whether their background knowledge is sufficient so that the quality of the signal they receive increases with extra effort (\emph{well-informed} agents) or worsens (\emph{misinformed} agents). Each agent $i$ is associated with an \emph{effort function} $f_i:\R_{\geq 0}\rightarrow [0,1]$ which relates the quality of the signal received by an agent with the effort she exerts as follows: the signal agent $i$ gets when she exerts an effort $x\geq 0$ is for the ground truth alternative $T$ with probability $f_i(x)$ and for alternative $F$ with probability $1-f_i(x)$. We assume that effort functions are continuously differentiable and have $f_i(0)=1/2$. The effort function for a well-informed agent $i$ is strictly increasing and strictly concave. The effort function for a misinformed agent is strictly decreasing and strictly convex. The functions $f_i(x)=1-\frac{e^{-x}}{2}$ and $f_i(x)=\frac{e^{-x}}{2}$ are typical examples of effort functions for a well-informed and a misinformed agent, respectively.

Agents are rational. They are involved in a \emph{strategic game} where they aim to maximize their 
\emph{utility}, consisting only of the payment they receive minus the effort they exert. In particular, we assume the agents are entirely indifferent to the outcome. This may lead to voting differently than what their signal indicates. We denote by $(\lambda_i,\beta_i)$ the \emph{strategy} of agent $i$, where $\lambda_i$ is the effort put and $\beta_i$ is the probability of casting a vote that is identical to the signal received (and, thus, the agent casts a vote for the opposite alternative with probability $1-\beta_i$). The utility of an agent is \emph{quasilinear}, i.e., equal to the amount of payments received minus the effort exerted. We assume that agents are \emph{risk neutral} and thus aim to maximize the expectation of their utility. Denote by $m_i$ the random variable indicating the number of agents different than $i$ who cast a $T$-vote. Clearly,  $m_i$ depends on the strategies of all agents besides $i$ but, for simplicity, we have removed this dependency from our notation. Now, the expected utility of agent $i$ when using strategy $(\lambda_i,\beta_i)$ is
\begin{align}\nonumber
&\E[u_i(\lambda_i,\beta_i,m_i)]\\\nonumber
&= -\lambda_i+f_i(\lambda_i)\beta_i\cdot \E\left[p\left(\frac{1+m_i}{n}\right)\right]\\\nonumber
&\quad\,+f_i(\lambda_i)(1-\beta_i)\cdot \E\left[p\left(\frac{n-m_i}{n}\right)\right]\\\nonumber
& \quad\, +(1-f_i(\lambda_i))\beta_i\cdot \E\left[p\left(\frac{n-m_i}{n}\right)\right]\\\nonumber
& \quad\, +(1-f_i(\lambda_i))(1-\beta_i)\cdot \E\left[p\left(\frac{1+m_i}{n}\right)\right]\\\nonumber
&= -\lambda_i+\E\left[p\left(\frac{1+m_i}{n}\right)\right]\\\label{eq:utility}
&\quad\,+\left(\beta_i(2f_i(\lambda_i)-1)-f_i(\lambda_i)\right)\cdot Q(m_i).
\end{align}
The quantities $p\left(\frac{1+m_i}n\right)$ and $p\left(\frac{n-m_i}n\right)$ are the payments agent $i$ receives when she votes for alternatives $T$ and $F$, respectively. The four positive terms in the RHS of the first equality above are the expected payments for the four cases defined depending on the signal received and whether it is cast as a vote or not. In the second equality, we have used the abbreviation 
\begin{align*}
Q(m_i)&=\E\left[p\left(\frac{1+m_i}{n}\right)-p\left(\frac{n-m_i}{n}\right)\right],
\end{align*}
which we also use extensively in the following. Intuitively, given the strategies of the other agents, $Q(m_i)$ is the additional expected payment agent $i$ gets when casting a $T$-vote compared to an $F$-vote.

We say that a set of strategies, in which agent $i\in [n]$ uses strategy $(\lambda_i,\beta_i)$, is an \emph{equilibrium} in the strategic game induced, if no agent can increase her utility by unilaterally changing her strategy. In other words, the quantity $\E[u_i(x,y,m_i)]$ is maximized with respect to $x$ and $y$ by setting $x=\lambda_i$ and $y=\beta_i$ for $i\in [n]$.

\section{Equilibrium analysis}
We are now ready to characterize equilibria. We remark that the cases (a), (b), and (c) of \cref{lem:equilibrium} correspond to the informal terms no correlation, positive correlation, and negative correlation used in the introductory section.
\begin{lemma}[equilibrium conditions]\label{lem:equilibrium}
The strategy of agent $i$ at equilibrium is as follows:
\begin{itemize}
\item[(a)] If $|f'_i(0) \cdot Q(m_i)| \leq 1$, then $\lambda_i=0$ and $\beta_i$ can have any value in $[0,1]$.
\item[(b)] If $f'_i(0) \cdot Q(m_i)> 1$, then $\lambda_i$ is positive and such that $f'_i(\lambda_i)\cdot Q(m_i)=1$ and $\beta_i=1$.
\item[(c)] If $f'_i(0) \cdot Q(m_i)<-1$, then $\lambda_i$ is positive and such that $f'_i(\lambda_i)\cdot Q(m_i)=-1$ and $\beta_i=0$.
\end{itemize}
\end{lemma}

\begin{proof}
First, observe that when agent $i$ selects $\lambda_i=0$, her expected utility is 
\begin{align*}
\E(u_i(0,\beta_i,m_i)]=\E\left[p\left(\frac{1+m_i}{n}\right)\right]-\frac{1}{2}Q(m_i),
\end{align*} i.e., it is independent of $\beta_i$. So, $\beta_i$ can take any value in $[0,1]$ when $\lambda_i=0$.

In case (b), we have $f'_i(0)\cdot Q(m_i)>0$ which, by the definition of the effort function $f_i$, implies that $(2f_i(\lambda_i)-1)\cdot Q(m_i)>0$ for $\lambda_i>0$. By inspecting the dependence of expected utility on $\beta_i$ at the RHS of equation (\ref{eq:utility}), we get that if agent $i$ selects $\lambda_i>0$, she must also select $\beta_i=1$ to maximize her expected utility in this case. Similarly, in case (c), we have $f'_i(0)\cdot Q(m_i)<0$ which implies that $(2f_i(\lambda_i)-1)\cdot Q(m_i)<0$ for $\lambda_i>0$. In this case, if agent $i$ selects $\lambda_i>0$, she will also select $\beta_i=0$ to maximize her expected utility. 

So, in the following, it suffices to reason only about the value of $\lambda_i$. Let 
\begin{align}\label{eq:deriv}
\nonumber\Delta_i(\lambda_i)&=\frac{\partial \E[u_i(\lambda_i,\beta_i,m_i)]}{\partial \lambda_i}
\\&= -1+(2\beta_i-1)f'(\lambda_i)\cdot Q(m_i)
\end{align}
denote the derivative of the expected utility of agent $i$ with respect to $\lambda_i$. In case (a), by the strict concavity/convexity of the effort function $f_i$ we have $|f'_i(\lambda_i) \cdot Q(m_i)| < 1$ for $\lambda_i>0$ and
\begin{align*}
    \Delta_i(\lambda_i) &= -1+(2\beta_i-1)f'_i(\lambda_i)\cdot Q(m_i)\\&\leq -1+|2\beta_i-1|\cdot |f'_i(\lambda_i)\cdot Q(m_i)| < 0.
\end{align*}
Hence, the expected utility of agent $i$ strictly decreases with $\lambda_i>0$ and the best strategy for agent $i$ is to set $\lambda_i=0$.

Otherwise, in cases (b) and (c), the derivative $\Delta_i(\lambda_i)$ has strictly positive values for $\lambda_i$ arbitrarily close to $0$ (this follows by the facts that $f$ is strictly convex/concave and continuously differentiable), while it is clearly negative as $\lambda_i$ approaches infinity (where the derivative of $f$ approaches $0$). Hence, the value of $\lambda_i$ selected by agent $i$ at equilibrium is one that nullifies the RHS of (\ref{eq:deriv}), i.e., such that $f'_i(\lambda_i)\cdot Q(m_i)=1$ in case (b) and $f'_i(\lambda_i)\cdot Q(m_i)=-1$ in case (c). Recall that $\beta_i$ is equal to $1$ and $0$ in these two cases, respectively.
\end{proof}

Using \cref{lem:equilibrium}, we can now identify some properties about the structure of equilibria.

\begin{lemma}\label{lemma:no_effort}
For any payment function, no effort by all agents (i.e., $\lambda_i=0$ for $i\in [n]$) is an equilibrium.
\end{lemma}

\begin{proof}
Notice that, when no agent puts any effort, each vote selects one of the two alternatives equiprobably. Then, the probability that $m_i$ takes a value $t\in \{0, 1, ..., n-1\}$ is equal to the probability that it takes value $n-1-t$. Hence, $\E\left[p\left(\frac{1+m_i}{n}\right)\right] = \E\left[p\left(\frac{n-m_i}{n}\right)\right]$ and $Q(m_i)=0$. Hence, all agents' strategies satisfy the condition of case (a) of Lemma~\ref{lem:equilibrium} and, thus, $\lambda_i=0$ is the best-response for each agent $i.$
\end{proof}
We will use the term \emph{non-trivial} for equilibria having at least one agent putting some effort.

The next lemma reveals the challenge of adjudication in our strategic environment. It essentially states that for every equilibrium that yields probably correct adjudication, there is an equilibrium that yields probably incorrect adjudication with the same probability.
\begin{lemma}\label{lem:mirror}
For any payment function, if the set of strategies $(\lambda_i,\beta_i)_{i\in [n]}$ is an equilibrium, so is the set of strategies $(\lambda_i,1-\beta_i)_{i\in [n]}$.
\end{lemma}

\begin{proof}
With a slight abuse of notation, we reserve the notation $m_i$ for the initial equilibrium where agent $i$ follows strategy $(\lambda_i,\beta_i)_{i\in [n]}$ and denote by $m'_i$ the random variable indicating the number of agents different than $i$ who cast a $T$-vote in the state where agent $i$ follows strategy $(\lambda_i,1-\beta_i)_{i\in [n]}$. Notice that, due to symmetry, the probability that $m_i$ gets a given value $t$ is equal to the probability that $m'_i$ gets the value $n-1-t$. Hence,
\begin{align*} &\E\left[p\left(\frac{1+m'_i}{n}\right)\right] = \E\left[p\left(\frac{n-m_i}{n}\right)\right] \\\mbox{ and  }
    &\E\left[p\left(\frac{n-m'_i}{n}\right)\right] = \E\left[p\left(\frac{1+m_i}{n}\right)\right].
\end{align*}
Thus, $Q(m_i)=-Q(m'_i)$, hence $f'_i(0)\cdot Q(m_i)=-f'_i(0)\cdot Q(m'_i)$. By \cref{lem:equilibrium}, we have that the strategies of all agents in the new state are consistent with the equilibrium conditions of \cref{lem:equilibrium}, provided that the initial state is an equilibrium (and thus satisfies the conditions).
\end{proof}

We say that an equilibrium is \emph{simple} if there exists an alternative $a\in \{T,F\}$ such that all agents cast a vote for alternative $a$ with probability at least $1/2$. Intuitively, this makes prediction of the agents' behaviour at equilibrium easy. Together with \cref{lem:equilibrium}, this definition implies that, in a simple equilibrium, an agent putting no effort (i.e., $\lambda_i=0$) can use any strategy $\beta_i$. For agents putting some effort, a well-informed agent uses $\beta_i=1$ if $a=T$ and $\beta_i=0$ if $a=F$ and a misinformed agent uses $\beta_i=0$ if $a=T$ and $\beta_i=1$ if $a=T$.

\begin{lemma}[simple equilibrium condition]\label{lem:simple}
When the payment function $p$ satisfies
\begin{align}\nonumber
&p\left(\frac{2+m}{n}\right)-p\left(\frac{1+m}{n}\right)\\\label{eq:lem-simple}
&\quad\, +p\left(\frac{n-m}{n}\right)-p\left(\frac{n-m-1}{n}\right) \geq 0,
\end{align}
for every $m\in \{0, 1, ..., n-2\}$, all equilibria are simple.
\end{lemma}
\begin{proof}
For the sake of contradiction, let us assume that the payment function $p$ satisfies the condition of the lemma but, at some equilibrium, agents $1$ and $2$ cast a $T$-vote with probability higher than $1/2$ and lower than $1/2$, respectively. Clearly, the equilibrium strategies of agents $1$ and $2$ cannot belong to case (a) of \cref{lem:equilibrium} as the probability of casting a $T$-vote would be exactly $1/2$ in that case.

We first focus on agent $1$ and distinguish between two cases. If her strategy is $\beta_1=1$, then it belongs to case (b) of \cref{lem:equilibrium} and, thus, $f'_1(0)\cdot Q(m_{1})>1$. Furthermore, the probability of casting a $T$-vote is $f_1(\lambda)$. Hence, $f_1(\lambda)>1/2$, implying that agent $1$ is well-informed with $f'_1(0)>0$. By the inequality above, we conclude that $Q(m_{1})>0$. If instead, agent $1$'s strategy is $\beta_1=0$, then it belongs to case (c) of \cref{lem:equilibrium} and, thus, $f'_1(0)\cdot Q(m_{1})<1$. The probability of casting a $T$-vote is now $1-f_1(\lambda)$. Hence, $f_1(\lambda)<1/2$, implying that agent $1$ is misinformed with $f'_1(0)<0$. By the inequality involving $Q(m_{1})$, we conclude that $Q(m_{1})>0$ again.

Applying the same reasoning for agent 2, we can show that $Q(m_{2})<0$. Hence,
\begin{align}\label{eq:pos-Q}
    Q(m_{1})-Q(m_{2})>0.
\end{align}

Denote by $X_1$ and $X_2$ the random variables indicating that agents $1$ and $2$ cast a $T$-vote and by $m$ the number of $T$-votes by agents different than $1$ and $2$. Let $\delta_i = \Pr[X_i=1]$. For $i\in \{1,2\}$, we have \begin{align}\nonumber
    Q(m_{3-i})&= \E\left[p\left(\frac{1+m+X_i}{n}\right)-p\left(\frac{n-m-X_i}{n}\right)\right]\\\nonumber
    &= \delta_i\cdot \E\left[p\left(\frac{2+m}{n}\right)-p\left(\frac{n-m-1}{n}\right)\right]\\\nonumber
    &\quad\, +(1-\delta_i)\cdot \E\left[p\left(\frac{1+m}{n}\right)-p\left(\frac{n-m}{n}\right)\right]\\\label{eq:Q-X_i}
    &= Q(m)+\delta_i\left(Q(m+1)-Q(m)\right).
\end{align}
Hence, from (\ref{eq:pos-Q}) and (\ref{eq:Q-X_i}) we obtain that 
\begin{align}\label{eq:dQ}
    \left(Q(m+1)-Q(m)\right)\cdot (\delta_2-\delta_1) &>0.
\end{align}
Notice that the assumption on $p$ implies that 
\begin{align*}
    &Q(m+1)-Q(m)\\ &= \E\left[p\left(\frac{2+m}{n}\right)-p\left(\frac{n-m-1}{n}\right)\right]\\\nonumber&\quad\,-\E\left[p\left(\frac{1+m}{n}\right)-p\left(\frac{n-m}{n}\right)\right]\geq 0,
\end{align*}
while our assumption on the probability of casting a $T$-vote implies $\delta_1>1/2>\delta_2$. These last two inequalities contradict (\ref{eq:dQ}) and the proof is complete.
\end{proof}
It can be verified that the payment function 
\begin{align*}
    p(x) &= \begin{cases} \frac{\omega}{xn}, & x\geq1/2\\ -\frac{\ell}{xn}, & x<1/2 \end{cases}
\end{align*}
with $\omega\leq \ell$ satisfies the condition of \cref{lem:simple}. We refer to this function as the award/loss sharing payment function. Essentially, the agents with the majority vote share an award of $\omega$ while the ones in minority share a loss of $\ell$. Note that for $\omega=\ell$, the payment function is strictly budget balanced unless all votes are unanimous. This is similar to the payment function used in Kleros. A sufficient condition for simple equilibria which is quite broad but does not include Kleros' payments is the following.

\begin{cor}\label{cor:simple}
When the payment functions are monotone non-decreasing, all equilibria are simple.
\end{cor}
 
\section{Selecting payments for correct adjudication}\label{sec:lp}
We now focus on the very simple scenario in which some of the $n$ agents are well-informed and have the same effort function $f$ and the rest are misinformed and have the effort function $1-f$. Can we motivate an expected $x$-fraction of them vote for the ground truth? 

Of course, we are interested in values of $x$ that are higher than $1/2$. This goal is directly related to asking for a high probability of correct adjudication. Indeed, as the agents cast their votes independently, the realized number of $T$-votes is sharply concentrated around their expectation and thus the probability of incorrect adjudication is exponentially small in terms of the number of agents $n$ and the quantity $(x-1/2)^2$. This can be proved formally by a simple application of well-known concentration bounds, e.g., Hoeffding's inequality~\cite{H63}.

So, our aim here is to define appropriate payment functions so that a set of strategies leading to an expected $x$-fraction of $T$-votes is an equilibrium. We will restrict our attention to payments satisfying the condition of \cref{lem:simple}; then, we know that all equilibria are simple. We will furthermore show that all equilibria are {\em symmetric}, in the sense that all agents cast a $T$-vote with the same probability. This means that there are $\lambda>0$ and $\beta\in \{0,1\}$ so that all well-informed agents use strategy $(\lambda,\beta)$ and all misinformed agents use the strategy $(\lambda,1-\beta)$.

\begin{lemma}\label{lem:symmetric}
Consider the scenario with $n$ agents, among which the well-informed agents use the same effort function $f$ and the misinformed agents use the effort function $1-f$. If the payment function $p$ satisfies the condition of \cref{lem:simple}, then all equilibria are symmetric.
\end{lemma}

\begin{proof}
For the sake of contradiction, assume that non-symmetric equilibra exist. Then, by Lemma~\ref{lem:mirror}, there exists an equilibrium, in which the agent $i$ putting the highest effort $\lambda_i>0$ is either well-informed and follows the strategy $(\lambda_i,1)$ or misinformed and follows the strategy $(\lambda_i,0)$, casting a $T$-vote with probability $f(\lambda_i)>1/2$. Let $j$ be another agent using strategy $(\lambda_j,\beta_j)$ with $\lambda_j<\lambda_i$. Since agent $i$ casts a $T$-vote with probability higher than $1/2$, agent $j$ is either well-informed and uses $\beta_j=1$ or misinformed and uses $\beta_j=0$; in any other case, she would cast a $T$-vote with probability less then $1/2$, contradicting the simplicity of equilibria from Lemma~\ref{lem:simple}. In both cases, the probability of casting a $T$-vote is
\begin{align}\label{eq:ineq-on-lambdas}
    f(\lambda_j)<f(\lambda_i).
\end{align}
Now, denote by $m$ the random variable indicating the number of agents different than $i$ and $j$ who case a $T$-vote. Then, it is $m_i=m+1$ with probability $f(\lambda_j)$ and $m_i=m$ with probability $1-f(\lambda_j)$. Thus, by the definition of $Q$, we get
\begin{align}\nonumber
    Q(m_i) &= \E\left[p\left(\frac{2+m}{n}\right)\right]\cdot f(\lambda_j)\\\nonumber
    & \quad + \E\left[p\left(\frac{1+m}{n}\right)\right]\cdot (1-f(\lambda_j))\\\nonumber
    & \quad - \E\left[p\left(\frac{n-m-1}{n}\right)\right]\cdot f(\lambda_j)\\\nonumber
    & \quad - \E\left[p\left(\frac{n-m}{n}\right)\right]\cdot (1-f(\lambda_j))\\\nonumber
    &= Q(m)+f(\lambda_j)\cdot 
\E\left[p\left(\frac{2+m}{n}\right)-p\left(\frac{1+m}{n}\right)\right.\\\label{eq:expression-for-Q}
&\quad\, \left.+p\left(\frac{n-m}{n}\right)-p\left(\frac{n-m-1}{n}\right)\right],
\end{align}
and an analogous equality for $Q(m_j)$. Since, by Lemma~\ref{lem:simple}, the expectation is non-negative, (\ref{eq:ineq-on-lambdas}) implies that 
\begin{align}\label{eq:ineq-on-Qs}
    Q(m_i)&\leq Q(m_j)
\end{align}
Now, by the equilibrium condition for agent $i$, we have $f'(\lambda_i)\cdot Q(m_i)=1$ (notice that this condition holds, no matter whether agent $i$ is well-informed or misinformed) and, hence,
\begin{align}\label{eq:positive-Q}
Q(m_i)&>0.
\end{align}
By the struct concavity of the effort function $f$ and since $\lambda_j<\lambda_i$, we also have that
\begin{align}\label{eq:ineq-on-derivs}
f'(\lambda_i)&<f'(\lambda_j).
\end{align}
Using the equilibrium condition for agent $j$ (again, this holds no matter whether agent $j$ is well-informed or misinformed) and inequalities (\ref{eq:ineq-on-Qs}), (\ref{eq:positive-Q}), and (\ref{eq:ineq-on-derivs}), we obtain 
\begin{align*}
    f'(\lambda_i)\cdot Q(m_i) &<f'(\lambda_j)\cdot Q(m_j)\leq f'(\lambda_j)\cdot Q(m_j)=1,
\end{align*}
which contradicts the equilibrium condition for agent $i$.
\end{proof}

 \cref{lem:symmetric} implies that, for $x>1/2$, an equilibrium with an expected $x$-fraction of $T$-votes has each agent casting a $T$-vote with probability $f(\lambda)=x$; the well-informed agents use the strategy $(\lambda,1)$ and the misinformed agents use the strategy $(\lambda,0)$. As agents vote independently, the random variables $m_i$ follow the same binomial distribution $\text{Bin}(n-1,x)$ with $n-1$ trials, each having success probability $x$. Also, notice that the fact that the effort function is strictly monotone implies that $\lambda$ is uniquely defined from $x$ as $\lambda=f^{-1}(x)$.

We now aim to solve the optimization task of selecting a payment function $p$ which satisfies the conditions of Lemma~\ref{lem:simple}, induces as equilibrium the strategy $(\lambda,1)$ for well-informed agents and the strategy $(\lambda,0)$ for misinformed agents, ensures non-negative expected utility for all agents (individual rationality), and minimizes the expected amount given to the agents as payment. As all agents cast a $T$-vote with the same probability and the quantities $m_i$ are identically distributed for different $i$s, it suffices to minimize the expected payment
\begin{align}\label{eq:obj}
    x\cdot \E\left[p\left(\frac{1+m_i}n\right)\right]+(1-x)\cdot \E\left[p\left(\frac{n-m_i}n\right)\right]
\end{align}
of a single agent. By the definition of expected utility in equation (\ref{eq:utility}), restricting this quantity to values at least as high as $f^{-1}(x)$ gives the individual rationality constraints for all agents. Furthermore, by Lemma~\ref{lem:equilibrium}, the equation,
\begin{align}
    \label{eq:minimal_eq_condition}
    f'(f^{-1}(x))\cdot Q(m_i)&=1,
\end{align}
gives the equilibrium condition for both well-informed and misinformed agents.

We can solve the optimization task above using linear programming. Our LP has the payment parameters $p(1/n)$, $p(2/n)$, ..., $p(1)$ as variables. The linear inequalities (\ref{eq:lem-simple}) for $m\in \{0, 1, ..., n-2\}$ form the first set of constraints, restricting the search to payment functions satisfying the conditions of Lemma~\ref{lem:simple}. Crucially, observe that the quantities $\E\left[p\left(\frac{1+m_i}n\right)\right]$ and $\E\left[p\left(\frac{n-m_i}n\right)\right]$ and, subsequently, $Q(m_i)$, can be expressed as linear functions of the payment parameters. Indeed, for $t=0, 1, ..., n-1$, let $z(t)=\Pr[m_i=t]$ be the known probabilities of the binomial distribution $\text{Bin}(n-1,x)$. Clearly,
\begin{align*}
\E\left[p\left(\frac{1+m_i}n\right)\right]&=\sum_{t=0}^{n-1}{z(t)\cdot p\left(\frac{1+t}n\right)},
\end{align*}
and,
\begin{align*}
\E\left[p\left(\frac{n-m_i}n\right)\right]&=\sum_{t=0}^{n-1}{z(t)\cdot p\left(\frac{n-t}n\right)}.
\end{align*}
Thus, the objective function (\ref{eq:obj}), the individual rationality constraint, and the equilibrium condition constraint can be expressed as linear functions of the LP variables. Overall, the LP has $n$ variables and $n+1$ constraints ($n$ inequalities and one equality).
The next statement summarizes the above discussion.
\begin{theorem}\label{thm:lp}
Consider the scenario with $n$ agents, among which the well-informed ones have the same effort function $f$ and the misinformed ones have the same effort function $1-f$. Given $x\in (1/2,1)$, selecting the payment function that satisfies the conditions of \cref{lem:simple}, induces an equilibrium in which all agents have non-negative expected utility and an expected $x$-fraction of agents casts a $T$-vote so that the expected amount given to the agents as payment is minimized, can be done in time polynomial in $n$ using linear programming.
\end{theorem}
Our approach can be extended to include additional constraints (e.g., non-negativity or monotonicity of payments), provided they can be expressed as linear constraints of the payment parameters. \cref{fig:minimal_payments} depicts four payment solutions obtained by solving the above LP for $n=100$ and the effort function $f(x)=1-\frac{e^{-x}}{2}$, and values of $x$ ranging from $51\%$ to $99\%$.

\definecolor{c1}{RGB}{68,1,84}
\definecolor{c4}{RGB}{94,201,98}
\definecolor{c3}{RGB}{31,163,134}
\definecolor{c2}{RGB}{55,89,140}
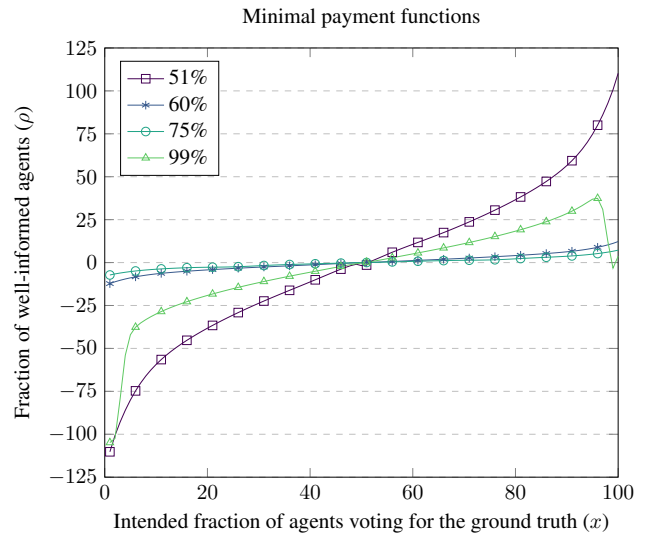
\begin{figure}
    \centering
    \resizebox{\columnwidth}{!}{
    \begin{tikzpicture}
    \begin{axis}[
        title={Minimal payment functions},
        xlabel={Intended fraction of agents voting for the ground truth ($x$)},
        ylabel={Fraction of well-informed agents ($\rho$)},
        xmin=0, xmax=100,mark repeat=5,
        ymin=-125, ymax=125,
        xtick={0,20,40,60,80,100},
        ytick={-150,-125,-100,-75,-50,-25,0,25,50,75,100,125,150},
        legend pos=north west,
        ymajorgrids=true,
        grid style=dashed,
    ]
    
    \addplot[
        color=c1,
        mark=square,
        ]
        coordinates {
        (1,-110.22327420537121)(2,-100.69195466897838)(3,-92.85301415026142)(4,-86.03560986138305)(5,-80.02197823010563)(6,-74.73251217802999)(7,-70.10484139879352)(8,-66.05925223813041)(9,-62.5021308247686)(10,-59.3417134131892)(11,-56.49933834891782)(12,-53.911976755962115)(13,-51.53025331361291)(14,-49.31555857230739)(15,-47.2375774997781)(16,-45.27243507218232)(17,-43.40131355485299)(18,-41.609385841788615)(19,-39.88497119345874)(20,-38.21885997594903)(21,-36.60377488373137)(22,-35.03394315986368)(23,-33.504758729282244)(24,-32.01251607862062)(25,-30.55420019027523)(26,-29.127318811855236)(27,-27.729764678555988)(28,-26.359696013178255)(29,-25.01542404975631)(30,-23.695297345500503)(31,-22.397575805676937)(32,-21.12029458461162)(33,-19.86113057581363)(34,-18.61730040513836)(35,-17.385531508117147)(36,-16.162143686068152)(37,-14.943244046079887)(38,-13.724976755489728)(39,-12.503723217922914)(40,-11.276199782340328)(41,-10.039603076576359)(42,-8.792223807386467)(43,-7.535073416910471)(44,-6.274884164045414)(45,-5.028130024551615)(46,-3.823607974998989)(47,-2.6984647091558642)(48,-1.6970071268126325)(49,-0.9907721838511669)(50,-1.7633572871764187)(51,-1.4450364017122688)(52,0.9051210563598069)(53,2.1817666359564214)(54,3.4561384655278538)(55,4.746271373334757)(56,5.999857426446411)(57,7.204675704554919)(58,8.369209536034024)(59,9.506057445477577)(60,10.626743137670207)(61,11.741083559567105)(62,12.857460814888032)(63,13.982921190096725)(64,15.123021302670953)(65,16.281684055880508)(66,17.461295162708304)(67,18.66305893663042)(68,19.88747172251245)(69,21.1347530976209)(70,22.405146433332355)(71,23.699080297572316)(72,25.017228561695127)(73,26.360517959421372)(74,27.73012231196051)(75,29.127467545003046)(76,30.554259378790462)(77,32.012538662697324)(78,33.50476702090202)(79,35.033946107142135)(80,36.60377590903587)(81,38.218860331992204)(82,39.884971321661524)(83,41.609385893680496)(84,43.40131358402531)(85,45.27243517748965)(86,47.23757755495858)(87,49.31555857244379)(88,51.53025331361848)(89,53.91197675596303)(90,56.49933834891788)(91,59.3417134131891)(92,62.5021308247686)(93,66.05925223813041)(94,70.10484139879352)(95,74.73251217802999)(96,80.02197823010563)(97,86.03560986138305)(98,92.85301415026142)(99,100.69195466897838)(100,110.22327420537121)
        };

    \addplot[
        color=c2,
        mark=asterisk,
        ]
        coordinates {
        (1,-12.238781672144164)(2,-11.182595079659224)(3,-10.313677257633579)(4,-9.557715724457646)(5,-8.890489477282248)(6,-8.303104990722215)(7,-7.788735349937528)(8,-7.338869174585101)(9,-6.943570548342665)(10,-6.592957437627303)(11,-6.2783405567597015)(12,-5.992592689208363)(13,-5.730044890256201)(14,-5.486237454938764)(15,-5.257673601424043)(16,-5.041614616964308)(17,-4.8359174278567645)(18,-4.638907369507594)(19,-4.449279213917474)(20,-4.266021163002234)(21,-4.088357774099843)(22,-3.9157083669224244)(23,-3.7476573420188117)(24,-3.583932067341152)(25,-3.4243828267392)(26,-3.268958384352075)(27,-3.117670879181779)(28,-2.970545782749674)(29,-2.8275566309301885)(30,-2.688549459973339)(31,-2.553167661446509)(32,-2.420795063457262)(33,-2.290544237186786)(34,-2.1613220165578344)(35,-2.0319867828317313)(36,-1.9015604664685184)(37,-1.7694043344773709)(38,-1.6352772591516476)(39,-1.4992771687390913)(40,-1.3617306664345268)(41,-1.223080726511201)(42,-1.0837919634401563)(43,-0.944294863599092)(44,-0.8049826609890509)(45,-0.6662861205597519)(46,-0.5289379336161022)(47,-0.39462125796505854)(48,-0.2671973685031359)(49,-0.1549358242779415)(50,-0.07733344057469438)(51,0.04452803628554669)(52,0.1445308145203219)(53,0.2571141780598012)(54,0.3808864282836444)(55,0.5112268901842412)(56,0.6444631493518993)(57,0.77806919465784)(58,0.910535790542931)(59,1.0411894572345823)(60,1.16998417614493)(61,1.2973263315551122)(62,1.4239495460854577)(63,1.5507992040365446)(64,1.6788976200847348)(65,1.8091875191261924)(66,1.942365341023792)(67,2.0787576457637678)(68,2.218315048816333)(69,2.360734917517571)(70,2.5056394163818614)(71,2.652721269346104)(72,2.801818780405575)(73,2.952929867485959)(74,3.1061909211973977)(75,3.2618418051649654)(76,3.4201908276465005)(77,3.5815885625273434)(78,3.7464156088114233)(79,3.9150857247145296)(80,4.088062896490703)(81,4.265889635593474)(82,4.449224205386964)(83,4.638885959401128)(84,4.835909773161072)(85,5.041612158117415)(86,5.257672916495795)(87,5.486237297742125)(88,5.730044862558652)(89,5.992592685951413)(90,6.278340556672721)(91,6.592957437700678)(92,6.943570548366998)(93,7.338869174590307)(94,7.78873534993847)(95,8.303104990722328)(96,8.890489477282328)(97,9.557715724457633)(98,10.313677257633579)(99,11.182595079659224)(100,12.238781672144164)
        };
        \addplot[
        color=c3,
        mark=o,
        ]
        coordinates {
        (1,-7.2286748093054145)(2,-6.603183444803278)(3,-6.088658736103817)(4,-5.64098132548435)(5,-5.245286205701512)(6,-4.895887775284087)(7,-4.588707345601121)(8,-4.319153076815775)(9,-4.082109012913337)(10,-3.8724742024352836)(11,-3.685650769335644)(12,-3.517853475578965)(13,-3.366288328575671)(14,-3.229259265872922)(15,-3.1062031143769744)(16,-2.9975813230421506)(17,-2.9045163804108074)(18,-2.828093525437776)(19,-2.768379378322357)(20,-2.723410559111354)(21,-2.6885675061282543)(22,-2.6567177054735085)(23,-2.619236926666735)(24,-2.5676617570792963)(25,-2.495531000668314)(26,-2.399958500022927)(27,-2.282301906560087)(28,-2.1474257367766443)(29,-2.002177432589992)(30,-1.8538553020196855)(31,-1.7088838763891414)(32,-1.57188215780315)(33,-1.4453068267416422)(34,-1.329652151661739)(35,-1.223997604359739)(36,-1.1266445752232404)(37,-1.035652868963121)(38,-0.9491950768239825)(39,-0.8657326554859508)(40,-0.7840620373128464)(41,-0.7032893383141468)(42,-0.6227837362389452)(43,-0.542145826260497)(44,-0.46121649090344863)(45,-0.3801477919831546)(46,-0.29956094642444286)(47,-0.22082620116376295)(48,-0.146509570170535)(49,-0.0810243806375488)(50,-0.031465399096098245)(51,0.03146425684421872)(52,0.08102239478145523)(53,0.14650451045583868)(54,0.22081283963416043)(55,0.29952670311147145)(56,0.3800634701387946)(57,0.46101690891800184)(58,0.541691211913107)(59,0.6217868443269856)(60,0.7011857404263144)(61,0.7797936121733506)(62,0.8574109850577725)(63,0.9336191889301226)(64,1.0076840508539402)(65,1.0784991991167927)(66,1.1446104031409934)(67,1.2043749646714905)(68,1.2563029480149337)(69,1.2995847541631826)(70,1.3347253162198)(71,1.3640985626402928)(72,1.3921664574286876)(73,1.4251554703723186)(74,1.470171885033122)(75,1.5339346018770428)(76,1.6213042061104057)(77,1.7338147053178852)(78,1.8689774676310544)(79,2.0209697649686946)(80,2.1822061239454236)(81,2.345164168358311)(82,2.5040199957659572)(83,2.6556650305341147)(84,2.799870719937661)(85,2.9387190504216747)(86,3.075687580082109)(87,3.214803891823929)(88,3.3601181193373897)(89,3.515536335205088)(90,3.6849189038235775)(91,3.872299343270762)(92,4.082089427794985)(93,4.319160937455089)(94,4.588713710454397)(95,4.895890506507115)(96,5.2452871462584065)(97,5.640981657809046)(98,6.088659095769978)(99,6.6031843008603355)(100,7.228674809364355)
        };
        
        \addplot[
        color=c4,
        mark=triangle,
        ]
        coordinates {
        (1,-104.85157324210462)(2,-102.39883028863952)(3,-79.24234923157131)(4,-54.02968128337184)(5,-42.01057831303496)(6,-37.63707787256836)(7,-35.17208937178961)(8,-33.215247163145904)(9,-31.49680965006507)(10,-29.947408601103163)(11,-28.53321247127379)(12,-27.23010781279506)(13,-26.01904846625105)(14,-24.884650689324722)(15,-23.81452328392572)(16,-22.7985917040225)(17,-21.828826947975543)(18,-20.89880116614618)(19,-20.00339352345211)(20,-19.138527631158965)(21,-18.300953144376578)(22,-17.48806684902013)(23,-16.697767939771175)(24,-15.928342082302436)(25,-15.17836915604691)(26,-14.446650104449924)(27,-13.732148929311156)(28,-13.033946462829297)(29,-12.351203089921038)(30,-11.683128059805398)(31,-11.028953424503351)(32,-10.38791099039255)(33,-9.759210995455932)(34,-9.142021569587724)(35,-8.535448455212304)(36,-7.938515043154391)(37,-7.350143637920857)(38,-6.769140199336983)(39,-6.19418692377762)(40,-5.623850450476754)(41,-5.056619125067035)(42,-4.4909922918838685)(43,-3.925661082282126)(44,-3.3598490916992887)(45,-2.7939315448963544)(46,-2.2305336577580803)(47,-1.6764219279088532)(48,-1.1455838782236931)(49,-0.6636974840892922)(50,-0.2731862269512195)(51,0.2731862269512195)(52,0.6636974840892922)(53,1.1455838782236931)(54,1.6764219279088532)(55,2.2305336577580803)(56,2.7939315448963544)(57,3.3598490916992887)(58,3.925661082282126)(59,4.4909922918838685)(60,5.056619125067035)(61,5.623850450476754)(62,6.19418692377762)(63,6.769140199336983)(64,7.350143637920857)(65,7.938515043154391)(66,8.535448455212304)(67,9.142021569587724)(68,9.759210995455932)(69,10.38791099039255)(70,11.028953424503351)(71,11.683128059805398)(72,12.351203089921038)(73,13.033946462829297)(74,13.732148929311156)(75,14.446650104449924)(76,15.17836915604691)(77,15.928342082302436)(78,16.697767939771175)(79,17.48806684902013)(80,18.300953144376578)(81,19.138527631158965)(82,20.003393523452132)(83,20.898801166146058)(84,21.828826947971535)(85,22.79859170395426)(86,23.814523295156214)(87,24.8846837666281)(88,26.019051410110524)(89,27.23010201471311)(90,28.533167156202786)(91,29.947101686863725)(92,31.494689063380086)(93,33.200260482205394)(94,35.06575380518058)(95,36.90527129794111)(96,37.41210920409582)(97,30.96315504142699)(98,12.204903117428765)(99,-3.2642757510561466)(100,3.712364404143216)
        };
        \legend{51\%,60\%,75\%,99\%};
        
    \end{axis}
    \end{tikzpicture}
    }
    \caption{Minimal payment functions that ensure the existence of a simple equilibrium with an $x$-fraction of the agents casting a $T$-vote on average, so that all agents have non-negative expected utility. The scenario uses $n=100$ and the effort function $f(x)=1-\frac{e^{-x}}{2}$. The payment functions obtained by solving the linear program from \cref{thm:lp} for $x\in\{0.51,0.6,0.75,0.99\}$ are shown. Each point $(x,y)$ on a curve means that an agent will receive a payment of $y$ if an $x$-fraction of the agents voted in the same way as they did. There is a marker for every fifth data point.}
    \label{fig:minimal_payments}
\end{figure}

\section{Computational experiments}\label{sec:experiments}
Our goal in this section is to justify that appropriate selection of the payment parameters can lead to correct adjudication in practice, even though Lemma~\ref{lem:mirror} shows the co-existence of both good and bad equilibria. The key property that favours good equilibria more often is that, in practice, jurors are on average closer to being well-informed than misinformed. Formally, this means that
$\frac{1}{n}\cdot \sum_{i\in [n]}{f_i(x)} > 1/2$ for every $x>0$.

Due to the lack of initial feedback, it is natural to assume that agents start their interaction by putting some small effort and convert their signal to a vote. We claim that this, together with their tendency to being well-informed, is enough to lead to probably correct adjudication despite strategic behaviour. We provide evidence for this claim through the following experiment implementing the scenario we considered in Section~\ref{sec:lp}.

We have $n$ agents, a $\rho$-fraction of whom are well-informed and the rest are misinformed. Agent $i$'s effort function is $f_i(x)=1-\frac{e^{-x}}{2}$ if she is well-informed and $f_i(x)=\frac{e^{-x}}{2}$ if she is misinformed. We consider the minimal payment functions, defined as the solution of the linear program detailed in the last section, parameterized by the fraction $x$ of agents intended to vote for the ground truth. A small subset of these payment functions can be seen in \cref{fig:minimal_payments}. In addition, we consider two different payment functions, both defined using a parameter $\omega>0$:
\begin{itemize}
    \item $p(x)=\omega$ if $x\geq 1/2$ and $p(x)=0$, otherwise. 
    \item $p(x)=\frac{\omega}{xn}$ if $x\geq 1/2$ and $p(x)=-\frac{\omega}{xn}$, otherwise.
\end{itemize}
With the first payment function, each agent gets a payment of $\omega$ if her vote is in the majority, while she gets no payment otherwise. With the second payment, the agents in the majority share an award of $\omega$, while the agents in the minority share a loss of $\omega$. Notice that both payment functions satisfy the conditions of \cref{lem:simple}. We will refer to them as {\em threshold} and {\em award/loss sharing} payment functions, respectively. 

In our experiments, we simulate the following dynamics of strategic play. Initially, all agents put an effort of $\epsilon>0$ and cast the signal they receive as vote. In subsequent rounds, each agent best-responds. In particular, the structure of the dynamics is as follows:
\begin{description}
\item[Round $0$:] Agent $i$ puts an effort of $\epsilon$ and casts her signal as vote.
\item[Round $j$, for $j=1, 2, ..., R$:] Agent $i$ gets $m_i$ as feedback. She decides her strategy $\beta_i\in \{0,1\}$ and effort level $\lambda_i\geq 0$. She draws her signal, which is alternative $T$ with probability $f_i(\lambda_i)$ and alternative $F$ with probability $1-f_i(\lambda_i)$. If $\beta_i=1$, she casts her signal as vote; otherwise, she casts the opposite of her signal as vote.
\end{description}
\definecolor{c1}{RGB}{68,1,84}
\definecolor{c4}{RGB}{94,201,98}
\definecolor{c3}{RGB}{31,163,134}
\definecolor{c2}{RGB}{55,89,140}
\begin{figure}[!ht]
    \centering
    \resizebox{\columnwidth}{!}{
    \begin{tikzpicture}
    \begin{axis}[
        title={Minimal payment functions},
        xlabel={Intended fraction of agents voting for the ground truth ($x$)},
        ylabel={Fraction of well-informed agents ($\rho$)},
        xmin=0, xmax=100,mark repeat=5,
        ymin=-160, ymax=160,
        xtick={0,20,40,60,80,100},
        ytick={-150,-125,-100,-75,-50,-25,0,25,50,75,100,125,150},
        legend pos=north west,
        ymajorgrids=true,
        grid style=dashed,
    ]
    
    \addplot[
        color=c1,
        mark=square,
        ]
        coordinates {
        (1,-150.97243507880225)(2,-137.97765668787252)(3,-127.29449686513627)(4,-118.00770731736328)(5,-109.82128378530712)(6,-102.62680445517914)(7,-96.33805720864052)(8,-90.84506613073006)(9,-86.01933582547417)(10,-81.73486330199)(11,-77.8831528741156)(12,-74.3770181423543)(13,-71.14840003416847)(14,-68.14456996934646)(15,-65.32456078789977)(16,-62.656375389922864)(17,-60.114926286257464)(18,-57.680501812111245)(19,-55.33759347904997)(20,-53.0739793921418)(21,-50.88000214272252)(22,-48.74799875958888)(23,-46.67185137339206)(24,-44.646633383685156)(25,-42.66833015625748)(26,-40.73361638416151)(27,-38.839674534381544)(28,-36.984040675180566)(29,-35.164466282380076)(30,-33.37878883346043)(31,-31.624812035758616)(32,-29.900209750709745)(33,-28.202484580075694)(34,-26.52902497748125)(35,-24.877297064184734)(36,-23.245155838018174)(37,-21.631148353459473)(38,-20.03452549712653)(39,-18.45456201115038)(40,-16.88886201883411)(41,-15.330765931887782)(42,-13.76681894753499)(43,-12.176297116462642)(44,-10.53539134326223)(45,-8.827519983952747)(46,-7.056895083588479)(47,-5.257710222767649)(48,-3.5038443379028976)(49,-2.002627290123975)(50,-1.8017597266871732)(51,0.3628713398617407)(52,2.6591543904426374)(53,4.080390295397784)(54,5.556687430351019)(55,7.122975836998215)(56,8.741302552426426)(57,10.389695816858545)(58,12.051748195219705)(59,13.710513982138991)(60,15.352289668085692)(61,16.970992527347942)(62,18.568879356754195)(63,20.15416816333934)(64,21.737590095758826)(65,23.32949712749765)(66,24.938236204469128)(67,26.56973388213725)(68,28.227852679679867)(69,29.915046186656145)(70,31.632990703246723)(71,33.383051772642986)(72,35.16657217425001)(73,36.98502850268639)(74,38.84011525528933)(75,40.73380371905405)(76,42.66840617741313)(77,44.64666291917726)(78,46.6718624072224)(79,48.74800275096066)(80,50.880003557542935)(81,53.0739798941364)(82,55.33759366468011)(83,57.6805018899106)(84,60.11492633199248)(85,62.656375562578376)(86,65.32456087895659)(87,68.14456996955622)(88,71.14840003417596)(89,74.37701814235555)(90,77.8831528741151)(91,81.73486330198999)(92,86.01933582547417)(93,90.84506613073006)(94,96.33805720864052)(95,102.62680445517914)(96,109.82128378530712)(97,118.00770731736328)(98,127.29449686513627)(99,137.97765668787252)(100,150.97243507880225)
        };

    \addplot[
        color=c2,
        mark=asterisk,
        ]
        coordinates {
        (1,-17.638251836557885)(2,-16.121027873558916)(3,-14.872826302121595)(4,-13.78696603036125)(5,-12.829020650623443)(6,-11.986538509469066)(7,-11.249790487554137)(8,-10.606258418870258)(9,-10.041106645443097)(10,-9.539634998660254)(11,-9.089115802682446)(12,-8.67929495899194)(13,-8.302147015209137)(14,-7.951437068264946)(15,-7.622944957275934)(16,-7.311634994425404)(17,-7.014574723365306)(18,-6.730560694615913)(19,-6.457288605506992)(20,-6.193370719317643)(21,-5.937787084351617)(22,-5.689832746424948)(23,-5.44909198890363)(24,-5.215423432449667)(25,-4.988942140399699)(26,-4.769981690393513)(27,-4.559017900510375)(28,-4.3565393507739305)(29,-4.162859718790514)(30,-3.9778839918927495)(31,-3.8008655857553157)(32,-3.630224486142486)(33,-3.463528753673655)(34,-3.297738723438575)(35,-3.1297192189712737)(36,-2.9568280725307257)(37,-2.7772131624767042)(38,-2.5895336752594744)(39,-2.392343118373793)(40,-2.1841647304902345)(41,-1.965122004019706)(42,-1.7384190750212878)(43,-1.5090075375250156)(44,-1.2810006054868586)(45,-1.0567903324668433)(46,-0.837888203115976)(47,-0.6265590666354313)(48,-0.4281093824814972)(49,-0.25449636864333236)(50,-0.13445202946415513)(51,0.051201677033446025)(52,0.20049017048337614)(53,0.36677636772304734)(54,0.5470382000318503)(55,0.7334614117920362)(56,0.9196629117945925)(57,1.101279232267978)(58,1.2762343534202159)(59,1.445139337394521)(60,1.611427965501981)(61,1.780095548033323)(62,1.9549170934014461)(63,2.137001393425232)(64,2.326181332001374)(65,2.5225496938457885)(66,2.7264297087523586)(67,2.937648978718839)(68,3.155159084250629)(69,3.3772513088737086)(70,3.6020942001230045)(71,3.82822943931616)(72,4.0548321286860585)(73,4.2817336163391095)(74,4.509296684321534)(75,4.73823642419409)(76,4.969450362652877)(77,5.2038915938439985)(78,5.442496852970018)(79,5.686166908651636)(80,5.935788486483254)(81,6.192284571548481)(82,6.456681930035689)(83,6.730188489876335)(84,7.014271655399645)(85,7.310276999729954)(86,7.621660043285573)(87,7.951432267432226)(88,8.30214671251358)(89,8.679294935617207)(90,9.08911580120332)(91,9.539634998668735)(92,10.041106645470924)(93,10.606258418877191)(94,11.24979048755542)(95,11.986538509469266)(96,12.829020650623558)(97,13.786966030361256)(98,14.872826302121595)(99,16.121027873558916)(100,17.638251836557885)
        };
        \addplot[
        color=c3,
        mark=o,
        ]
        coordinates {
        (1,-11.228542174366048)(2,-10.25716972637916)(3,-9.45816648330616)(4,-8.763101839056253)(5,-8.149311496883616)(6,-7.608343733785631)(7,-7.133969568597111)(8,-6.718784696550329)(9,-6.354323569843088)(10,-6.032258443515705)(11,-5.745598807145573)(12,-5.489585043265706)(13,-5.262244517702698)(14,-5.064517928654642)(15,-4.8998379746502625)(16,-4.773030092805748)(17,-4.688381627443128)(18,-4.646941146401543)(19,-4.643675623200379)(20,-4.665756833654683)(21,-4.693320062736779)(22,-4.702974629088165)(23,-4.672724333332662)(24,-4.586358412866636)(25,-4.436222163042021)(26,-4.224237422368109)(27,-3.961223201252979)(28,-3.6645038333468394)(29,-3.354170141851678)(30,-3.0489879082910667)(31,-2.7631831431833374)(32,-2.5049343125419448)(33,-2.276641942017889)(34,-2.07638768125038)(35,-1.8997779745043006)(36,-1.7415571988475431)(37,-1.5967054018880802)(38,-1.4610037742356088)(39,-1.3311950602111366)(40,-1.2049049111200207)(41,-1.080466854195671)(42,-0.956748296154819)(43,-0.833034062369129)(44,-0.7089991056963254)(45,-0.584795343778302)(46,-0.4612869199996239)(47,-0.340488663852764)(48,-0.2262808450539553)(49,-0.12545024379998804)(50,-0.0489632565262057)(51,0.04896170185202209)(52,0.1254475318833359)(53,0.22627394085477492)(54,0.34047042030954033)(55,0.4612399107993075)(56,0.5846782593532462)(57,0.7087171901947982)(58,0.8323776753601777)(59,0.9552716101578453)(60,1.0772612344789305)(61,1.1981998401536842)(62,1.3176989187807533)(63,1.434889683097993)(64,1.548174951614198)(65,1.6550060810193647)(66,1.7517696445702873)(67,1.8339199664936845)(68,1.8965212121467978)(69,1.935327014177441)(70,1.94838420507465)(71,1.937882155054644)(72,1.9116413392878657)(73,1.8834414618682223)(74,1.8715929426615743)(75,1.8958126820923265)(76,1.973241903240733)(77,2.114824666727266)(78,2.32300391439043)(79,2.591077910523566)(80,2.904207611946724)(81,3.2421439569886616)(82,3.583570518724848)(83,3.9109962069618827)(84,4.214280150804361)(85,4.49146196414307)(86,4.747184918236724)(87,4.990057299810407)(88,5.230222012389103)(89,5.477758872462305)(90,5.741972654990803)(91,6.031381953321605)(92,6.354177485335943)(93,6.718779402196558)(94,7.133977077352362)(95,7.608347548998266)(96,8.149312839785406)(97,8.76310231565919)(98,9.458166994493396)(99,10.257170859062288)(100,11.22854217445663)
        };
        
        \addplot[
        color=c4,
        mark=triangle,
        ]
        coordinates {
        (1,-153.29159432553112)(2,-104.69575039271423)(3,-126.3427215540618)(4,-72.92570211571606)(5,-53.451065015315585)(6,-47.85058670592181)(7,-44.59173231069414)(8,-41.949412408642345)(9,-39.647175769227196)(10,-37.607687826165275)(11,-35.7832092160096)(12,-34.13452955828363)(13,-32.628800500222795)(14,-31.23904547740878)(15,-29.94350804912895)(16,-28.724860199721515)(17,-27.569406986899608)(18,-26.466366672509736)(19,-25.407257600230608)(20,-24.385395431656605)(21,-23.395490822021156)(22,-22.43333247082128)(23,-21.4955398076667)(24,-20.579370980226315)(25,-19.682573963452672)(26,-18.803270829789568)(27,-17.93986720967802)(28,-17.090980632667478)(29,-16.255382772711954)(30,-15.431951671331174)(31,-14.619630841310336)(32,-13.81739283143481)(33,-13.024205437610329)(34,-12.238999370784116)(35,-11.460636958104372)(36,-10.687882528386703)(37,-9.919376758943518)(38,-9.1536197984469)(39,-8.388971970713484)(40,-7.623687126162078)(41,-6.856003490257473)(42,-6.084332080076127)(43,-5.3076063741360215)(44,-4.525893389625406)(45,-3.7414218858218593)(46,-2.9602646815154863)(47,-2.195017985025819)(48,-1.4689149017529015)(49,-0.8217547179792284)(50,-0.31751949109781563)(51,0.31751949109781563)(52,0.8217547179792284)(53,1.4689149017529015)(54,2.195017985025819)(55,2.9602646815154863)(56,3.7414218858218593)(57,4.525893389625406)(58,5.3076063741360215)(59,6.084332080076127)(60,6.856003490257473)(61,7.623687126162078)(62,8.388971970713484)(63,9.1536197984469)(64,9.919376758943518)(65,10.687882528386703)(66,11.460636958104372)(67,12.238999370784116)(68,13.024205437610329)(69,13.81739283143481)(70,14.619630841310336)(71,15.431951671331174)(72,16.255382772711954)(73,17.090980632667478)(74,17.93986720967802)(75,18.803270829789568)(76,19.682573963452672)(77,20.579370980226315)(78,21.4955398076667)(79,22.43333247082128)(80,23.395490822021156)(81,24.385395431656605)(82,25.407257600230615)(83,26.46636667250963)(84,27.569406986899086)(85,28.724860199709212)(86,29.943508048913397)(87,31.23904547390884)(88,32.62880044787207)(89,34.13452883890436)(90,35.78320016380814)(91,37.607583959191906)(92,39.646094037889796)(93,41.939214411527466)(94,44.50436606916391)(95,47.159685348851944)(96,48.26950897527338)(97,36.39918375859456)(98,-8.350271303945796)(99,23.288757780971114)(100,-13.27677003567071)
        };
        \legend{51\%,60\%,75\%,99\%};
        
    \end{axis}
    \end{tikzpicture}
    }
    \caption{Minimal payment functions computed using the approach of Section~\ref{sec:lp}, after relaxing \Cref{eq:minimal_eq_condition} to a lower bound inequality. The scenario uses $n=100$ and the effort function $f(x)=1-\frac{e^{-x}}{2}$. The payment functions obtained by solving the linear program from \cref{thm:lp} for $x\in\{0.51,0.6,0.75,0.99\}$ are shown. Each point $(x,y)$ on a curve means that an agent will receive a payment of $y$ if an $x$-fraction of the agents voted in the same way as they did. There is a marker for every fifth data point.}
    \label{fig:minimal_payments2}
\end{figure}
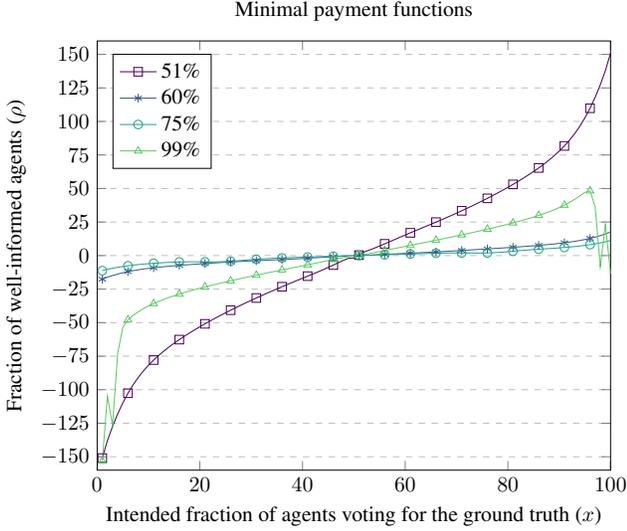
In each round after round $0$, agents get the exact value of $m_i$ as feedback (as opposed to its distribution)\footnote{An alternative implementation would assume that $m_i$ takes the number of $T$-votes in a randomly chosen previous round. The results obtained in this way are qualitatively similar to those we present here.} but maximize their expected utility with respect to the components $\lambda_i$ and $\beta_i$ of their strategy. Hence, the only difference with what we have seen in earlier sections is that the calculation of expected utility considers the actual value of payments and not their expectation, i.e.,
\begin{align*}
    \E[u_i(\lambda_i,\beta_i,m_i)]
    &= -\lambda_i+p\left(\frac{1+m_i}{n}\right)\\
    &\quad+(\beta_i(2f_i(\lambda_i)-1)-f_i(\lambda_i))\cdot Q(m_i),
\end{align*}
where 
\begin{align*}
    Q(m_i) &= p\left(\frac{1+m_i}n\right)-p\left(\frac{n-m_i}n\right).
\end{align*}
By applying \cref{lem:equilibrium}, we get the following characterization of the best-response of agent $i$ in round $j>0$.

\begin{cor}
    The best response of agent $i$ receiving feedback $m_i$ is as follows: 
    \begin{enumerate}
        \item[(a)] If $|Q(m_i)|\leq 2$, then $\lambda_i=0$ and $\beta_i$ can take any value in $[0,1]$. 
        \item[(b)] Otherwise, $\lambda_i=\ln{\frac{|Q(m_i)|}{2}}$.
        \begin{enumerate}
            \item[(b.1)] If agent $i$ is well-informed and $Q(m_i)>2$ or agent $i$ is misinformed and $Q(m_i)<-2$, then $\beta_i=1$.
            \item[(b.2)] If agent $i$ is misinformed and $Q(m_i)>2$ or agent $i$ is well-informed and $Q(m_i)<-2$, then $\beta_i=0$.
        \end{enumerate}
    \end{enumerate}
\end{cor}

\begin{figure*}[t]
    \centering
        \includegraphics[scale=0.2375]{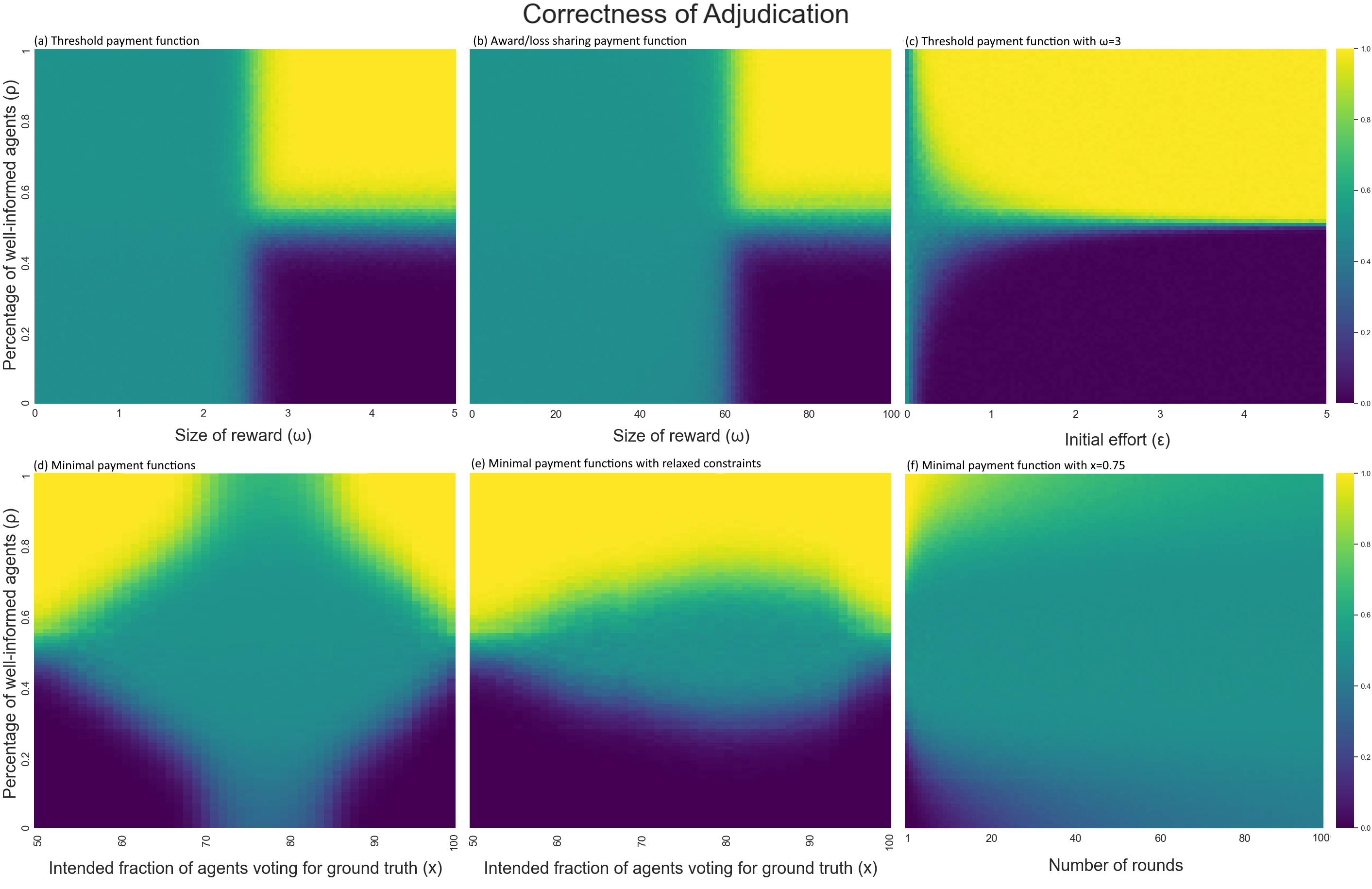} 
    \caption{Heatmap of the correctness of the adjudication, plotted with the fraction of well-informed agents on the y-axis, with six varying x-axes. In each plot, we run $R=50$ rounds with a jury of size $n=100$, using 1000 samples for each data point. The color of a data point indicates the average measured correctness with the given parameters, using the viridis color scale displayed in the legend on the right. Yellow  corresponds to good recovery, while dark blue corresponds to poor recovery of the ground truth, while random outcomes are represented by turquoise. The six x-axes are as follows: (a) Size of the reward for the threshold payment function, ranging from $\omega=0$ to $\omega=5$, with $\epsilon=1$. (b) Size of the reward for the award/loss sharing payment function, ranging from $\omega=0$ to $\omega=100$, with $\epsilon=1$. (c) The initial effort $\epsilon$, ranging from $\epsilon=0$ to $\epsilon=5$, with the payment function being the threshold payment function with $\omega=3$. (d) The intended fraction $x$ of agents voting for the ground truth, ranging from $x=0.51$ to $x=1$, with the payment function defined by \cref{thm:lp}.  (e) The intended fraction $x$ of agents voting for the ground truth, ranging from $x=0.51$ to $x=1$, with the payment functions obtained from \cref{thm:lp} by relaxing \cref{eq:minimal_eq_condition} to an inequality. (f) The number of rounds, ranging from $R=1$ to $R=100$, with the payment function being the minimal payment function with $x=0.75$ from \cref{fig:minimal_payments}.}
    \label{fig:payment1_rho_w}
\end{figure*}

In our experiments, we consider an agent population of fixed size $n=100$, with the fraction of well-informed agents ranging from $0$ to $1$. We simulate the dynamics described above for $R=50$ rounds and repeat each simulation $20$ times. For each experiment, we measure the frequency with which the majority of votes after the $R$-th round is for the ground truth alternative $T$. We do so for both the threshold and award/loss sharing payment functions, with parameter $\omega$ taking values ranging between $0$ and $5$ for the threshold payment functions and between $0$ and $100$ for the award/loss sharing one. We also consider the payment functions that arise as solutions to the linear programs considered in the previous section. In each experiment, we play with the values of two parameters simultaneously. We consider 100 values on each axis and plot the resulting data using a heatmap, with each data point corresponding to the average correctness observed during the experiment. We represent the correctness using the viridis color scale, with yellow points corresponding to a good recovery of the ground truth, and dark blue points corresponding to poor recovery. Random values are represented by turquoise points.

In the first experiment (\cref{fig:payment1_rho_w}.a), we consider the threshold payment function and vary the size of the reward $\omega$ and the fraction $\rho$ of well-informed agents. We consider a reasonably high starting effort of $\epsilon=1$, corresponding to a probability of $0.816$ of receiving the ground truth as signal. We observe two distinct regions as we vary the size of the payment. Initially, when the payment is too small (i.e. $\omega \leq 2.5$), the outcome of the adjudication is mostly random. When the payment increases above the threshold, we observe a {\em sharp phase transition} independent of $\rho$, where the correctness is extremified by the payment in the following sense: when $\rho$ is sufficiently large (respectively, small), the mechanism recovers the ground truth with high (respectively, low) probability. When $\rho\approx0.5$, we see that the outcome of the adjudication is mostly random. 

In the second experiment (\cref{fig:payment1_rho_w}.b), we consider the award/loss sharing payment function. The range of $\omega$ is changed from $[0,5]$ to $[0,100]$, as the latter constitutes the total award, while the former is the award per agent. All other parameters are kept the same. We obtain similar results as for the threshold payment function, i.e. the outcome is mostly random below a threshold above which we observe a sharp phase transition where the outcome of the mechanism is extremified. Here, the phase transitions happens when the total award is $\omega \approx 60$.

In the third experiment (\cref{fig:payment1_rho_w}.c), we observe the effect on the correctness by the initial effort. We fix the threshold payment function with $\omega=3$ such that mechanism has a chance to recover the ground truth, and let $\epsilon$ range from $0$ to $5$. We observe that, when $\epsilon$ is small, the outcome of the mechanism is mostly random, while the outcome quickly extremifies as $\epsilon$ increases. This means the mechanism only works if the agents initially put in sufficient effort. The results are similar for both the award/loss sharing payment function and the minimal payment functions.

In the fourth experiment (\cref{fig:payment1_rho_w}.d), we consider the payment functions obtained from \cref{thm:lp}. A subset of the payment functions we use are depicted in \cref{fig:minimal_payments}. Here, instead of varying the size of the reward, we vary the parameter $x$ used as input to the linear program. This parameter represents the intended fraction of agents voting for the ground truth at equilibrium. We let $x$ range from $0.51$ to $1$ in increments of $0.01$. Here, we observe that for $x$ close to $0.5$ and for $x$ close to 1, the mechanism is extremified, while for $x$ close to $0.75$ and $\rho$ close to $0.5$ the outcome of the mechanism is mostly random. This is rather unexpected since if a $0.75$-fraction of the agents vote for the ground truth, the majority vote will be for the ground truth almost certainly. Indeed, we observe that in these games when $\rho \approx 0.5$, the agents exert effort close to zero, hence producing the random outcome. We claim that despite this behavior, the ground truth is still an equilibrium, it is just not a stable equilibrium and the parties converge to the trivial equilibrium.

In a fifth experiment (\cref{fig:payment1_rho_w}.e), we consider a different set of minimal payment functions, obtained by relaxing the equality constraint \cref{eq:minimal_eq_condition} to a lower bound inequality. This has the effect of no longer requiring an exact $x$-fraction of the agents vote for the ground truth, but instead gives a lower bound on their number. This slightly changes the payment functions which can be seen in \cref{fig:minimal_payments2}, though they are qualitatively similar to those shown in \cref{fig:minimal_payments}. Here, we again vary the fraction $\rho$ of well-informed agents on the y-axis, and the intended fraction $x$ of agents voting for the ground truth, ranging from $x=0.51$ to $x=1$ in increments of $0.01$. However, we obtain different and considerably better results than those in \cref{fig:payment1_rho_w}.d. In particular, we obtain a good adjudication outcome for any $x$ when $\rho > 0.75$.

In our sixth and final experiment (\cref{fig:payment1_rho_w}.f), we aim to explain the enigmatic behaviour of the LP-computed payments for $x\approx 0.75$. We fix the payment function to be the minimal payment function with $x=0.75$ and vary the number of rounds from 1 round to 100 rounds. We do not take into account round 0 where all parties exert $\epsilon>0$ effort in the estimation of the correctness of the outcome. We observe that the outcome is extremified when the number of rounds is small and decays as we increase the number of round. We can explain this result by considering the payment function for $x=0.75$ in \cref{fig:minimal_payments} whose distribution is mostly flat when the outcome is close to being a tie. Here, the value of $Q(m_i)$ is small so the agent will lower the effort they exert, making it more likely that the outcome will be disputed. This creates a pull towards the trivial equilibrium. By contrast, the curves for $x\in\{0.51,0.99\}$ have a higher slope close to $0.5$, which makes this effect less pronounced. This explains why the adjudication outcome is mostly random for $x\approx 0.75$. By design, the linear program finds minimal payments that ensure there is an equilibrium where an $x$-fraction of the agents vote in favor of the ground truth. However, it does not constrain the solution to have the property that the good equilibrium is {\em stable}. In some sense, the fact that the non-trivial equilbrium is stable when $x$ is far from 0.5 is happenstance and begs the deeper question why the solutions to the linear program are of the form we observe. Intuitively, it makes sense that attaining a high accuracy requires large payments. A similar phenomenon seemingly holds for accuracies close to $0.51$ which can be explained informally as follows. Combinatorially, there are only a few ways to attain an accuracy of $0.51$ which necessitates the use of large punishment and rewards when the vote is close to being a tie. By contrast, for larger $\rho$, there are more ways to attain an accuracy of $0.75$ in the majority, hence loosening the requirements on the payments. This suggests that the case $x=0.75$ does not provide positive results in practice because of {\em instability of equilibria}. It would be interesting to explore whether it is possible to extend our approach with additional natural constraints that ensure the non-trivial equilibrium is also stable.

Our experiments suggest that several classes of payment functions can be used to recover the ground truth with high probability, provided the
agents are well-informed on average. Clearly, there is much work yet to be done in designing payment functions with desirable properties: while the threshold function and the award/loss sharing function seem to recover the ground truth reliably, it might be difficult in practice to pinpoint the location of the phase transition, as this requires estimating the effort functions used by actual jurors. The same holds true for the minimal payment functions.

\section*{Acknowledgments}
We would like to thank Luca Nizzardo, Irene Giacomelli, Matteo Campanelli, and William George for interesting discussions in several stages of this work. IC was partially supported by a research advisorship grant from Protocol Labs.

\newpage

\bibliographystyle{named}
\bibliography{refs}

\end{document}